\documentclass{lipics-v2019}
\usepackage{algorithmic}
\usepackage[ruled,vlined]{algorithm2e}
\usepackage{pagenote}
\usepackage{float}
\usepackage{tikz}
\usepackage{hyperref}

\renewcommand*{\pagenotesubhead}[1]{}

\usepackage{mdframed} %
\mdfdefinestyle{figstyle}{ %
  linecolor=black!7, %
  backgroundcolor=black!7, %
  innertopmargin=10pt, %
  innerleftmargin=25pt, %
  innerrightmargin=25pt, %
  innerbottommargin=10pt %
}

\newcommand{\ket}[1]{\ensuremath{|#1\rangle}}

\makepagenote
\author{Jackson Morris}{Department of Mathematics, University of California Los Angeles}{jrexmo@ucla.edu}{}{}
\author{Fang Song}{Department of Computer Science, Portland State
  University}{fsong@pdx.edu}{}{}{}
\begin{document}


\title{Simple vertex coloring algorithms}
\authorrunning{J. Morris and F. Song}

\maketitle
\begin{abstract}
  Given a graph $G$ with $n$ vertices and maximum degree $\Delta$, it
  is known that $G$ admits a vertex coloring with $\Delta + 1$ colors
  such that no edge of $G$ is monochromatic. This can be seen
  constructively by a simple greedy algorithm, which runs in time
  $O(n\Delta)$. 
  Very recently, a sequence of results (e.g., [Assadi et. al. SODA'19,
  Bera et. al. ICALP'20, AlonAssadi Approx/Random'20]) show randomized
  algorithms for $(\epsilon + 1)\Delta$-coloring in the query model
  making $\tilde{O}(n\sqrt{n})$ queries, improving over the greedy
  strategy on dense graphs.  In addition, a lower bound of
  $\Omega(n\sqrt n)$ for any $O(\Delta)$-coloring is established on
  general graphs.
  
  In this work, we give a simple algorithm for
  $(1 + \epsilon)\Delta$-coloring. This algorithm makes
  $O(\epsilon^{-1/2}n\sqrt{n})$ queries, which matches the best
  existing algorithms as well as the classical lower bound for
  sufficiently large $\epsilon$. Additionally, it can be readily
  adapted to a quantum query algorithm making
  $\tilde{O}(\epsilon^{-1}n^{4/3})$ queries, bypassing the classical
  lower bound.

  Complementary to these algorithmic results, we show a quantum lower
  bound of $\Omega(n)$ for $O(\Delta)$-coloring.
  
\end{abstract}
\newpage
\newcommand{\vc}{\text{vc}}

\section{Introduction}
\label{sec:intro}

Graph coloring is a fundamental problems in discrete algorithms and
graph theory. It has wide applications in scheduling, resource
allocation, compiler optimization, and logistics settings. In this
problem, one aims to assign every vertex a color such that no edge is
\emph{monochromatic}, i.e., both endpoints having the same
color. Finding such a valid coloring with $k \geq 1$ colors is usually
called the $k$-coloring problem. A critical graph parameter $\chi(G)$
is the minimum number of colors necessary to guarantee a valid
coloring. In Karp's renowned result, 3-coloring as well as deciding
$\chi(G)$ in general are proven NP-complete~\cite{Karp72}. On the
other hand, it is known that for any graph $G$,
$\chi(G) \leq \Delta + 1$ where $\Delta$ is the maximum degree of
$G$. This follows from a simple greedy algorithm for
$\Delta + 1$-coloring. The basic observation is that any valid partial
$\Delta + 1$-coloring of $G$ can be extended to a complete
$\Delta + 1$-coloring. This is because every vertex has at most
$\Delta$ neighbors, there must be at least one choice of color for
that vertex that does not conflict with any of its neighbors.

For quite some time, the $O(n\Delta)$ greedy algorithm is the best
known algorithm for $\Delta + 1$ vertex coloring on general
graphs. Very recently, authors of~\cite{ACK19} give a new
$\tilde{O}(n\sqrt{n})$ algorithm in the \emph{query} model, where the
graph's adjacency matrix is given as a black-box. This algorithm
relies on a tool called palette sparsification, in which a small
number of colors are sampled at every vertex and a valid coloring is
found from there. These tools are ingenious but at the same time
rather complicated, when compared to the vanilla greedy algorithm and
looking ahead the algorithms presented in this paper. Additionally,
the authors of \cite{alon2020palette} have extended this theory of
palette sparsification, achieving improved results for when $G$ is
triangle free and for $(1 + \epsilon)\Delta$-coloring in
general. Another important line of research concerns solving these
problems in a space-efficient manner: $\cite{ACK19, space-concious}$
develop algorithms for $\Delta + 1$-coloring and related problems
using a (near) linear amount of space. Meanwhile An almost matching
lower bound of $\Omega(n\sqrt{n})$ is also known, implying that their
algorithm is optimal up to $\text{polylog}(n)$
factors. 

\subparagraph{Our Contribution.} In this work, we revisit the graph
coloring problem, in both classical and quantum query models. 
We give a new randomized algorithm for $(1 + \epsilon)\Delta$ coloring
which is optimal for sufficiently large $\epsilon$, and a quantum
variant of this algorithm follows readily that bypasses the classical
lower bound in~\cite{ACK19}. The most appealing feature of our
algorithms is their simplicity. Basically, we are inspired by the
$\Delta+1$ greedy algorithm, where the algorithms progresses by
extending a partial solution \emph{locally}, and it is suitable to
speed up ``local'' search subroutine by quantum unstructured
search. In particular, we do not need the heavy machinery such as
palette sparsification and list-coloring in existing classical
algorithms. We overview our results below, and see also
Table~\ref{tab:sum} for a summary.

\begin{mdframed}[style=figstyle,innerleftmargin=10pt,innerrightmargin=10pt]
\textbf{Result 1}: A randomized algorithm for
$(1 + \epsilon) \Delta$-coloring which makes
$O(\epsilon^{-1/2}n\sqrt{n})$ queries in expectation and always returns
a valid coloring.
\end{mdframed}

This basic idea inherits the same general ``local search''
strategy. For each vertex, we sample a random color and check if it
conflicts with any of its neighbors that has been assigned
already. This simple approach is surprisingly
efficient 
We show that this algorithm has expected query complexity
$O(\frac{n^2}{\epsilon\Delta})$, and as in existing works we then just
need to choose between our algorithm and the vanilla greedy algorithm,
which gives the desired
$O(\min\{n\Delta, \frac{n^2}{\epsilon\Delta}\}) =
O(\epsilon^{-1/2}n\sqrt{n})$ expected query complexity. This is
optimal for $\epsilon^{-1/2} = O(1)$ as we show in
Corollary~\ref{cor:lvlb}, and its Monte Carlo variant improves upon
existing algorithms~\cite{alon2020palette} by a factor of
$\epsilon^{-1}$. This algorithm also serves a template for designing
our efficient quantum algorithm. Our new algorithm only needs oracle
access to the adjacency matrix (i.e., pair queries), and the vanilla
greedy algorithm needs access to the adjacency list (i.e., neighbor
queries). Therefore the final algorithm assumes both kinds of oracles
are available.

A central subroutine in our algorithm, which is again extremely
simple, is what we call \textbf{Find-Conflict}. It takes a vertex $v$,
a set of vertices $S$ such that $v \not \in S$ and uses the graph
oracle to determine if there are any edges of the form $(u, v)$
actually present in $E(G)$ for any of the $u \in S$. Classically, this
amounts to an exhaustive search on all such vertices $s \in S$.

\begin{mdframed}[style=figstyle,innerleftmargin=10pt,innerrightmargin=10pt]
\textbf{Result 2}: A quantum algorithm which makes $\tilde{O}(\epsilon^{-1}n^{4/3})$ queries and returns a valid $(1 + \epsilon)\Delta$-coloring with high probability. 
\end{mdframed}

Achieving this result essentially has the same strategy as that of
result 1, while noting that \textbf{Find-Conflict} can be instantiated
by (variants of) Grover's quantum search
algorithm~\cite{grover,Safe-Grover,Boyer_98}, under the standard
quantum query model where the adjacency matrix can be queried in
quantum superposition. The analysis nonetheless requires more care to
carefully contain the errors. 

We remark that although we rely on well-known quantum algorithmic
techniques, it is crucial that they get employed in the right
place. For instance, one may be tempted to perform a search over just
the color palette, thus requiring $O(\sqrt{\Delta}\log\Delta)$
queries, but such an approach would need a rather powerful ``color
validity'' oracle for every vertex, which needs to be dynamically
updated as various colors become forbidden for some
vertices. Maintaining these oracles from standard oracles (adjacency
matrix or adjacency list) would incur a significant overhead that
would wipe out possible quantum speedup.

\begin{mdframed}[style=figstyle,innerleftmargin=10pt,innerrightmargin=10pt]
\textbf{Result 3}: Quantum query lower bound of $\Omega(n)$ for $O(\Delta)$ coloring, when (quantum) pair queries are allowed.
\end{mdframed}

This is shown by concocting specific graph instances and reducing the
unstructured search problem to them, when only adjacency matrix queries are available. Therefore the $\Omega(\sqrt{n})$
quantum query lower bound for unstructured
search~\cite{grover-optimal} will transfer.

\begin{table}[H]
  \centering
  \begin{tabular}{ |c|c|c|c|  }    
 \hline
Problem & Quantum Algorithm & Classical Algorithm & \shortstack{Classical \\ Lower Bound}\\
 \hline\hline
 $2\Delta$-coloring   & $\tilde{O}(n^{4/3})$ & $\tilde{O}(n\sqrt{n})$, $\tilde{O}(n\sqrt{n})$ \cite{ACK19}&   $\Omega(n\sqrt{n})$ \cite{ACK19}\\ \hline
 $(1 + \epsilon)\Delta$-coloring &  $\tilde{O}(\epsilon^{-1}n^{4/3})$& ${O}(\epsilon^{-1/2}n\sqrt{n}\log{n})$, \newline${O}(\epsilon^{-3/2}n\sqrt{n\log{n}})$ \cite{alon2020palette}   &$\Omega(n\sqrt{n})$ \cite{ACK19}\\
 \hline
\end{tabular}
\caption{Summary of results - no reference indicates a result from
  this work. The $2\Delta$-coloring problem occurs when $\epsilon = 1$ and we include this case to demonstrate the quantum advantage. Since any $\Delta + 1$-coloring is also a $2\Delta$-coloring by definition, we use \cite{ACK19} for classical $2\Delta$-coloring, but this comparison is not of great importance.  While our randomized algorithm is given as a Las Vegas
  algorithm, it can easily be converted to Monte-Carlo with success
  probability at least $1 - 1/n^k$ for any $k$ using only
  $O(\epsilon^{-1/2}n\sqrt{n}\log{n})$ queries. Note that prior to
  this work the best know algorithm for
  $(1 + \epsilon)\Delta$-coloring from \cite{alon2020palette} has
  query complexity ${O}(\epsilon^{-3/2}n\sqrt{n\log{n}})$ compared to
  our algorithm's ${O}(\epsilon^{-1/2}n\sqrt{n}\log{n})$. Note that
  the former performs better for $\epsilon^{-1} \leq \sqrt{\log{n}}$,
  while the latter performs better for
  $\epsilon^{-1} \geq \sqrt{\log{n}}$.}
\label{tab:sum}
\end{table}
\subparagraph*{Further discussion.} One immediate problem left open by
our results in to close the gap between the algorithmic bounds and the
lower bound. 
Some other well studied graph problems that are likely to
admit similar quantum speedups are variants of coloring (defective
coloring, edge coloring, etc) and potentially some dynamic graph
problems.


A problem that has received considerable study in the quantum query
model is triangle finding. Here one aims to either output a triangle
if one exists or determine that the graph is triangle free with
bounded error (or other variants such as listing all triangles). For a general graph this may require $\Omega(n^2)$
queries to the adjacency matrix oracle. However, a modified Grover
search can be used, searching over triples of vertices, to achieve an
$\tilde{O}(n\sqrt{n})$ quantum algorithm for triangle finding as in
\cite{older_triangle}. More sophisticated arguments in
\cite{quantum_triangle} improve upon this resulting, bring the query
complexity down to $\tilde{O}(n^{5/4})$. Conversely, the best know
lower bound for triangle finding in the quantum query model is the
immediate $\Omega(n)$ bound established in \cite{non-adaptive}. Just
as in triangle finding, a gap persists between the best known quantum
lower bound and the best known quantum algorithm for $O(\Delta)$
coloring as shown in our work. 

In addition, the problem of \emph{maximum matching} in which one wishes to find the largest
maximal matching, has been studied in the quantum query model. The very recent results of
\cite{maximum-matching} gives an improved quantum algorithm for maximum matching in the adjacency
matrix model which makes $O(n^{7/4})$ queries. However, a lower bound of $\Omega(n^{3/2})$
queries in the adjacency matrix model established in \cite{graph_lb} show that a gap persists
between the best known algorithm and lower bound for this problem as well. Other structural graph
problems such as
those relating to connected components and spanning forests have also
been studied. A recent work~\cite{cut-query}, surprisingly, shows
exponential quantum speedup for these problem, assuming a more
sophisticated oracle model that can answer ``cut queries''.


\section{Preliminaries}
\label{sec:prelim}

In this paper the notation $[n]$ will refer to the set
$\{1, 2, \ldots n\}$ for any positive integer $n$. For a positive
integer $L$, the $L$-coloring problem is as follows: given a graph
$G = (V, E)$ and $n = |V|$ we wish to find an assignment
$c:[L] = \{1, \dots L\} \to V$ such that for any edge $(x, y) \in E$
we have $c(x) \not = c(y)$. $L$ will be referred to as the palette
throughout. Such an assignment is called a \emph{valid $L$-coloring of
  $G$}. Let $\Delta$ be the maximum degree of any vertex in $G$. As
previously stated a simple greedy algorithm can be shown to always
produce a valid coloring and runs in time $O(n\Delta)$. Rather than
time complexity, the results in this paper will mostly be stated in
terms of query complexity.  The main models of computation in this
paper are the graph query model and the quantum graph query model. In
the standard graph query model, we assume that $n$ and the degrees of
the vertices is the only knowledge available about $G$ before any
queries have been made. The standard query model for graphs supports
the following types of queries:
\begin{itemize}
\item Pair Queries: we can query the oracle if the edge $(v_i, v_j)$
  is present in the graph for any $i, j \in \{1, \dots n\}$. This will
  be denoted as \[ M[v_i, v_j] = \begin{cases}
      1 & (v_i, v_j) \in E\\
      0 & \text{otherwise}
   \end{cases}
 \]
 This is equivalent to granting oracle access to the adjacency matrix
 $(M_{i,j})$. 

\item Neighbor Queries: we can query the oracle for the $j$th neighbor of vertex $v_i$ for any $i\in \{1, \dots n\}$ and $j \in \{1, \dots \deg{v_i}\}$. When $j > \deg{v_i}$ the query returns $Null$.
\end{itemize}
Similarly this is equivalent to granting oracle access to the adjacency
list of the graph. 

The classic greedy $\Delta + 1$-coloring algorithm can easily be
implemented in the graph query model via neighbor queries. Explicitly,
we can determine which colors are valid at any vertex $v \in V$ with
$\deg{v}$ neighbor queries. 


In the quantum setting, we assume standard quantum query access to the
adjacency matrix. Namely we allow ``pair queries'' in
superposition. More precisely, we encode each vertex $i$ in a
$O(\log n)$-qubit register $\ket{i}$, and assume a black-box unitary
operation

\[ O_M: \ket{i,j} \mapsto (-1)^{M(i,j)}\ket{i,j}, \quad \forall i,j\in
  V(G)\, . \] Each application of this unitary transformation will be
referred to as a quantum query.


Quantum ``neighbor query'' oracle can be defined similarly (i.e.,
black-box unitary computing the adjacency list). However, all our
algorithms will only require \emph{classical} neighbor queries, and
hence we do not specify the quantum oracle explicitly.

Throughout this paper we use the phrase \textit{with high probability}
to mean with probability at least $1 - \frac{1}{n^k}$ for a
sufficiently large constant $k$.

\section{Algorithmic Results}
In this section we give a simple optimal randomized algorithm for $(1 + \epsilon)\Delta$-Coloring. This algorithm serves as a "warm-up" and as the inspiration for the quantum algorithms presented later. We start by transforming a bound on Monte Carlo algorithms for $O(\Delta)$-Coloring to one on Las Vegas algorithms for the same class of problems.
From the lower bounds established in \cite{ACK19}:
\begin{corollary}
Any randomized algorithm which always returns a valid $2\Delta$
coloring requires ${\Omega}(n\sqrt{n})$ queries in expectation.
\label{cor:lvlb}
\end{corollary}
\begin{proof}[Proof of corollary 1]
Lemma 5.6 of \cite{ACK19} implies that any Monte Carlo algorithm for $2\Delta$ coloring which makes fewer than $\frac{n\sqrt{n}}{400000}$ queries returns an invalid coloring with probability at least $\frac{1}{4}$. So, suppose that some Las Vegas algorithm which makes $O(\frac{n^{3/2}}{T(n)})$ queries in expectation for some $T: \mathbb{N} \rightarrow \mathbb{N}$ with $\lim_{n \to \infty}T(n) = \infty$. Now, consider a Monte Carlo variant of this algorithm which runs until a valid $2\Delta$-coloring is found or until $\frac{n\sqrt{n}}{500000}$ queries have been made. The probability that this Monte Carlo variant fails to return a valid coloring is equal to the probability that the Las Vegas algorithm makes $\frac{n\sqrt{n}}{500000}$ queries or more on the same input. Let $Q$ denote the number of queries made by the Las Vegas algorithm. From Markov's inequality, we have
\begin{align*}
  \mathbb{P}(Q \geq \frac{n\sqrt{n}}{500000}) \leq
  \frac{\mathbb{E}[Q]}{a}
  \leq\frac{n\sqrt{n}/T(n)}{n\sqrt{n}/500000}
  = \frac{500000}{T(n)}
\end{align*}
Since $T(n) \to \infty$ as $n \to \infty$, there must exist some $N$ such that for all $n > N$, $T(n) \geq 5000000$ and therefore for $n > N$
\begin{align*}
  \mathbb{P}(Q \geq \frac{n\sqrt{n}}{500000}) &\leq \frac{500000}{T(n)}\leq \frac{1}{10}
\end{align*}
Thus, this Monte-Carlo variant produces a valid $2\Delta$-coloring with probability at least $\frac{9}{10}$ and makes at most $\frac{n\sqrt{n}}{5000000}$ queries for sufficiently large $n$.
This contradicts lemma 5.6 of \cite{ACK19}, proving the stated lower bound.
\end{proof}
This result will guide us in our search for randomized algorithms. 
\subsection{Randomized $(1 + \epsilon)\Delta$-Coloring}
We begin this section by providing a Las Vegas algorithm for $(1 + \epsilon)\Delta$-coloring in the standard query model.
\begin{theorem}\label{theorem:theorem 2}
There exists a randomized algorithm which returns a valid $(1 + \epsilon)\Delta$ vertex coloring and makes $O(\epsilon^{-1/2}n\sqrt{n})$ queries in expectation.
\end{theorem}
\begin{proof}[Proof of theorem 2]
Recall that the greedy algorithm for $\Delta + 1$-coloring algorithm
makes $O(n\Delta)$ queries, so when $\Delta \ll n$ this algorithm
performs very well. This helps us on our search for sublinear
algorithms since we can use this just greedy algorithm when $\Delta$
is sufficiently small, allowing us to restrict our attention to larger
$\Delta$ i.e. $\Delta \geq \sqrt{n}$ - this is the case in previous work \cite{alon2020palette, ACK19}. Since $\Delta$ is not given, we must solve for it. This can be done with neighbor queries by binary searching at each vertex, i.e. we binary search at $v$ to find the maximum $i$ such that the neighbor query $(v, i)$ does not return $Null$. Doing this at every vertex and then taking the maximum uses $O(n\log{n})$ queries.

Both the classical and quantum coloring algorithms in this paper will have the same general strategy: when coloring a vertex choose a random color, test whether that color is a valid choice, if it is we move on, otherwise repeat. While this is a simple approach with some obvious limitations, it yields some surprisingly efficient algorithms. The only data structure we maintain in all of our algorithms is defined as follows:
\begin{itemize}
    \item For every $c\in [(1 + \epsilon)\Delta]$ we maintain a set called $\chi(c)$ of all vertices which we have been colored with color $c$. Let $\chi_t(c)$ denote the collection of vertices with color $c$ after having colored $t - 1$ vertices.
\end{itemize}
Our randomized algorithm is as follows:

\begin{mdframed}[style=figstyle,innerleftmargin=10pt,innerrightmargin=10pt]
\begin{algorithm}[H]
\SetAlgoLined
\textbf{for} $v \in V$\\
\ \ \ \ \textbf{while} $v$ is not colored\\
\ \ \ \ \ \ \ \ Choose a color $c$ uniformly at random from $[(1 + \epsilon)\Delta]$ \\
\ \ \ \ \ \ \ \ \textbf{for} $u \in \chi(c)$\\
\ \ \ \ \ \ \ \ \ \ \ \ Query $M[u, v]$\\
\ \ \ \ \ \ \ \ \ \ \ \ \textbf{if} $M[u, v] = 1$\\
\ \ \ \ \ \ \ \ \ \ \ \ \ \ \ \ \textbf{break} choose a new color\\
 \ \ \ \ \ \ \ \ Assign color $c$ to $v$ and update $\chi(c)$ with $v$\\
\textbf{return} The coloring assignment $\{\chi(1), \dots \chi((1 + \epsilon)\Delta)\}$
\caption{$(1 + \epsilon)\Delta$-Color($G(V, E), \epsilon$)}
\end{algorithm}
\end{mdframed}

Note that this algorithm always returns a valid $(1 + \epsilon)\Delta$-coloring since a color $c \in [(1 + \epsilon)\Delta]$ is only ever assigned to $v$ if no edge of the form $(u, v)$ for $u \in \chi(c)$ is present in the graph. Every vertex has at least one valid color since a vertex $v$ has at most $\Delta$ neighbors and therefore at least $(1 + \epsilon)\Delta - \Delta = \epsilon\Delta \geq 1$ valid colors.

Now, we will show that $(1 + \epsilon)\Delta$-Color makes $O(\epsilon^{-1/2}n\sqrt{n})$ queries in expectation, for sufficiently large $\Delta$. Let $Q(G)$ denote the number of queries made by this algorithm for a graph $G$ and let $Q_{t}$ denote the number of queries made when coloring vertex $v_t$. Observe that
\begin{align*}
    \mathbb{E}[Q(G)] &= \mathbb{E}\bigg{(}\sum_{t =
                       1}^{n}Q_{t}\bigg{)} = \sum_{t = 1}^n \mathbb{E}[Q_{t}]
\end{align*}
So, bounding $\mathbb{E}[Q_{t}]$ will allow us to bound the total expected number of queries.
\begin{lemma}\label{lemma:lemma 3}
For $\Delta \geq \sqrt{n/\epsilon}$ and any $t \leq n$$$ \mathbb{E}[Q_{t}] = O(\sqrt{n/\epsilon})$$
\end{lemma}
\begin{proof}[Proof of lemma 3] Let $c$ be the first color that is chosen for the vertex we are currently looking at. There are two possibilities: 1) $c$ is valid at $v_t$ or 2) it is not. Let $p = \mathbb{P}[c \text{ is valid at }v_t]$, we have:
 \begin{align*}
    \mathbb{E}[Q_t] &= p\mathbb{E}[|\chi(c)| \ : \ c \text{ is valid at }v_t] + (1 - p)(\mathbb{E}[|\chi(c)| \ : \ c \text{ is invalid at }v_t] + \mathbb{E}[Q_t])\\
    &= p\mathbb{E}[|\chi(c)| \ : \ c \text{ is valid at }v_t] + (1 - p)\mathbb{E}[|\chi(c)| \ : \ c \text{ is invalid at }v_t] + (1 - p) \mathbb{E}[Q_t]\\
    &= \mathbb{E}[|\chi(c)|] + (1 - p) \mathbb{E}[Q_t]
\end{align*}
When $c$ is chosen uniformly at random from $[(1 + \epsilon)\Delta]$, $\mathbb{E}[|\chi(c)|] \leq \frac{n}{(1 + \epsilon)\Delta}$. Also, recall that there are at least $\epsilon\Delta$ colors valid colors for $v_t$, so $p \geq \frac{\epsilon\Delta}{(1 + \epsilon)\Delta} = \frac{\epsilon}{1 + \epsilon}$. From these facts it follows that
\begin{align*}
  \mathbb{E}[Q_t] &= \frac{1}{p}\mathbb{E}[|\chi(c)|]
                    \leq \frac{1 + \epsilon}{\epsilon}\frac{n}{(1 + \epsilon)\Delta}
                    = \frac{n}{\epsilon\Delta}
\end{align*}
Hence, when $\Delta \geq \sqrt{n/\epsilon}$
\[
\mathbb{E}[Q_t] \leq \sqrt{n/\epsilon} \, . \]

\end{proof}
From Lemma~\ref{lemma:lemma 3} we have that
\begin{align*}
    \sum_{t = 1}^n \mathbb{E}[Q_t] = O(\epsilon^{-1/2}n\sqrt{n})
\end{align*}
when $\Delta \geq \sqrt{n/\epsilon}$.
So, $(1 + \epsilon)$-Color makes $O(\epsilon^{-1/2}n\sqrt{n})$ queries in expectation for $\Delta \geq \sqrt{n/\epsilon}$ and the greedy algorithm suffices for $\Delta \leq \sqrt{n/\epsilon}$. This concludes the proof of Theorem~\ref{theorem:theorem 2}.
\end{proof}
Whenever $\epsilon^{-1/2} = O(1)$ we have that $\mathbb{E}[Q(G)] =
O(n\sqrt{n})$. Hence this algorithm achieves the query lower bound for
Las Vegas algorithms established in corollary 1 for sufficiently large
values of $\epsilon$ - this is somewhat surprising given the exceedingly simple nature of this randomized algorithm. However, this procedure is not useful when $\epsilon^{-1}
= \Delta$ ($\Delta + 1$-coloring), so finding an $O(n\sqrt{n})$ algorithm (Las Vegas or Monte Carlo) for $\Delta +
1$ coloring remains as an open problem. Also, note that this algorithm can readily be converted to have a fixed query complexity with error bounded by $O(1/n^k)$ for any constant $k$ via Markov's inequality. Doing so incurs a multiplicative cost of an additional $O(\log{n}$  It would be interesting to see if the randomness used here could be combined with the palette sparsification technique of \cite{ACK19} to achieve an improved algorithm for $\Delta + 1$-coloring.
\\
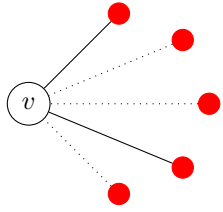
\begin{figure}
\centering 
\begin{tikzpicture}[transform shape]
      \node[draw, circle]
    (N0) at ({180}:1.2cm) {$v$};
    
    \node[draw, circle, inner sep=.1cm, red, fill]
    (N1) at ({90}:1.2cm) {};
    \node[draw, circle, inner sep=.1cm, red, fill]
    (N2) at ({45}:1.2cm) {};
    \node[draw, circle, inner sep=.1cm, red, fill]
    (N3) at ({0}:1.2cm) {};
    \node[draw, circle, inner sep=.1cm, red, fill]
    (N4) at ({315}:1.2cm) {};
    \node[draw, circle, inner sep=.1cm, red, fill]
    (N5) at ({270}:1.2cm) {};
    
   \path (N1) edge[thin] (N0);
   \path (N2) edge[dotted] (N0);
   \path (N3) edge[dotted] (N0);
   \path (N4) edge[thin] (N0);
   \path (N5) edge[dotted] (N0);

\end{tikzpicture}
\caption{Here, we query all edges whose existence would prevent us from coloring $v$ red}
\end{figure}
\section{Quantum Algorithms}
The simple randomized algorithm explored in the previous section involved repeatedly solving the following sub-problem:
Given a vertex $v \in V$ and a collection of vertices $S \subset V\setminus v$ determine if $\{(u, v) : u \in S\} \cap E(G) = \emptyset$. In other words, does $v$ have any neighbors in $S$? 
In effect, this is just an instance of unstructured search. However, we do not know how many (if any) such edges are actually present, so a standard Grover search algorithm will not work here. Luckily, a slight modification of Grover search resolves this problem.
\begin{claim}
There is a quantum algorithm which given a set of vertices $S \subset V$ and $v \in V \setminus S$  determines whether or not $v$ is adjacent to any vertex in $S$ and returns such a vertex in $S$ if it exists using $O(\sqrt{|S|})$ queries to the oracle with success probability at least $\frac{2}{3}$, without prior knowledge of $|\{ (v, s) \in E \ | \ s \in S\}|$.
\end{claim}
This is a restatement of Theorem 3 of \cite{tight_bounds}, whose algorithm is sketched in the appendix. This result was originally proven in \cite{Boyer_98}. This $2/3$ success probability can be amplified to at least $1 - \frac{1}{n^3}$ via $O(\log{n})$ additional runs. So, we will let $\textbf{Find-Conflict}(v, S)$ be the algorithm which determines if $v$ is adjacent to $s \in S$ with failure probability at most $\frac{1}{n^2}$ and uses $O(\sqrt{|S|}\log{n})$ queries. We can additionally guarantee that \textbf{Find-Conflict} does not return false positives (concluding there is a conflict when no such conflict exists) by having it return a proposed neighbor of $v$, $s \in S$, and then make the classical query $(v, s)$ to ensure this edge is actually present in the graph. Thus, the only chance of error is when we incorrectly conclude there is no conflict, which happens with probability at most $\frac{1}{n^3}$. A detailed version of this algorithm is included in Appendix \ref{appendix:Appendix A}.

In $(1 + \epsilon)\Delta$-Color we needed to spend $|\chi(c)|$ queries in the worst case in order to determine if $c$ is valid at $v$, but by using $\textbf{Find-Conflict}$ we are able to achieve a quantum speed up. Another change in the case of the quantum algorithm is that we wish only have one type of randomness (accuracy or query complexity). To simplify the analysis we only allow the algorithm to test "small" color classes; explicitly, we only look at color classes of size at most $\frac{2n}{\epsilon\Delta}$. We did not do this in the case of the randomized algorithm since the expected size of a uniformly random color class gives the desired asymptotic. Since we now wish to bound the error of the algorithm, we fix the total number of queries at $T = 9\epsilon^{-3/2}\log^2{n}\sqrt{\frac{n^3}{\Delta}}$.

\begin{theorem}\label{theorem:theorem 5}
There exists a quantum algorithm which returns a valid $(1 + \epsilon)\Delta-$vertex coloring with high probability and makes $\tilde{O}(\epsilon^{-1}n^{4/3})$ queries.
\end{theorem}
\begin{proof}[Proof of theorem 5]
Consider the algorithm below:
\begin{mdframed}[style=figstyle,innerleftmargin=10pt,innerrightmargin=10pt]
\begin{algorithm}[H]
\SetAlgoLined
$t \gets 0$\\
\textbf{while} $t \leq T$\\
\ \ \ \ \textbf{for} $v \in V$:\\
\ \ \ \ \ \ \ \ \textbf{while} $v$ is not colored \\
\ \ \ \ \ \ \ \ \ \ \ \ Choose a color $k$ uniformly at random from $[(1 + \epsilon)\Delta]$ \\
\ \ \ \ \ \ \ \ \ \ \ \ \textbf{if} $|\chi(k)| >  \frac{2n}{\epsilon\Delta}$\\
\ \ \ \ \ \ \ \ \ \ \ \ \ \ \ \ \textbf{break} choose a new color\\
\ \ \ \ \ \ \ \ \ \ \ \ \textbf{else}\\
\ \ \ \ \ \ \ \ \ \ \ \ \ \ \ \ \textbf{if} \textbf{Find-Conflict}($v, \chi(k)$)\\
\ \ \ \ \ \ \ \ \ \ \ \ \ \ \ \ \ \ \ \ $t \gets t + 2\sqrt{\frac{n}{\epsilon\Delta}}\log{n}$\\
\ \ \ \ \ \ \ \ \ \ \ \ \ \ \ \ \ \ \ \ \textbf{break} choose a new color\\
\ \ \ \ \ \ \ \ \ \ \ \ \ \ \ \ \textbf{else}\\
\ \ \ \ \ \ \ \ \ \ \ \ \ \ \ \ \ \ \ \ Assign $k$ to $v$ and update $\chi(k)$\\
\textbf{return} The coloring assignment $\{\chi(1), \dots \chi((1 + \epsilon)\Delta)\}$
\caption{Quantum-Color($G(V, E)$, $\epsilon$)}
\end{algorithm}
\end{mdframed}

\begin{lemma}\label{lemma:lemma 6}
Quantum-Color returns valid $(1 + \epsilon)\Delta$-coloring of $G$ with high probability.
\end{lemma}

\begin{proof}[Proof of lemma 6]
First, note that that in this algorithm \textbf{Find-Conflict} is only ever called when $|\chi(c)| \leq \frac{2n}{\epsilon\Delta}$, so every call makes at most $2\sqrt{\frac{n}{\epsilon\Delta}}\log{n}$ queries and has failure probability at most $1/n^2$. This algorithm fails to produce a valid coloring in two cases:
\\
\begin{enumerate}
    \item Some vertex is assigned an invalid color
    \item Not all of the vertices are colored by the time $9\epsilon^{-3/2}\log^2 n \sqrt{\frac{n^3}{\Delta}}$ queries have been made
\end{enumerate}
\vspace{3mm}
For the first case, note that any $v \in V$ is assigned an invalid color only when $\textbf{Find-Conflict}$ returns \textbf{false} when it should return \textbf{true}. This occurs with probability at most $\frac{1}{n^2}$ when a color which is invalid at $v$ is chosen. So a union bound shows that when all vertices are assigned a color, the coloring is valid with probability at least $1 - \frac{1}{n^2}$. 

Now, we will show that with high probability every vertex is assigned some color before using too many queries. Let $Q_t$ be the number of queries made when attempting to color vertex $v_t$. Additionally, let $K_t$ be the number of times that a small random color ($|\chi(k)| \leq \frac{2n}{\epsilon\Delta}$) is tested when attempting to color $v_t$. Note that $K_t$ is the number of times that \textbf{Find-Conflict} is called on $v_t$. Let $\ell = 9 \epsilon^{-3/2}\log^2 n\sqrt{\frac{n}{\Delta}}$. Then,
\begin{align*}
\mathbb{P}[Q_t \geq \ell] &\leq  \mathbb{P}[K_t \geq 9\epsilon^{-1}\log n]
\end{align*}
Since \textbf{Find-Conflict} cannot have false positives, $K_t$ only depends on the number of invalid colors that we test with \textbf{Find-Conflict}. In other words, we are guaranteed to stop once a valid small color is chosen and we will use this to upper bound $K_t$.
\\
\\
Now, note that there are at most $\frac{\epsilon\Delta}{2}$ colors $k$ such that $|\chi(k)| \geq \frac{2n}{\epsilon\Delta}$ and since there are at least $\epsilon\Delta$ valid colors for $v$, there must be at least $\frac{\epsilon\Delta}{2}$ colors which are both small and valid. For a color $k$ which is randomly chosen among all small colors, it is valid with probability at least $\frac{\epsilon\Delta/2}{(1 + \epsilon\Delta)} = \frac{\epsilon}{2(1 + \epsilon)} \geq \frac{\epsilon}{4}$. Thus,  
\begin{align*}
  \mathbb{P}[Q_t \geq \ell] &\leq  \mathbb{P}[K_t \geq 9\epsilon^{-1}\log n]
                              \leq \bigg{(}1 - \frac{\epsilon}{4} \bigg{)}^{9 \epsilon^{-1}\log{n}}
\end{align*}
Now, note that for $\epsilon > 0$, $\bigg{(}1 - \frac{\epsilon}{4}\bigg{)}^{\epsilon^{-1}} < \frac{1}{\sqrt[4]{e}}$, so 
\begin{align*}
  \mathbb{P}[Q_t \geq \ell] &\leq \bigg{(}\frac{1}{\sqrt[4]{e}} \bigg{)}^{9\log{n}}
                              \leq \bigg{(}\frac{1}{2}\bigg{)}^{3\log{n}}
                              = \frac{1}{n^3}
\end{align*}
By union bound, we get 
\begin{align*}
    \mathbb{P}\bigg{[} \sum_{t = 1}^n Q_t \geq 9\epsilon^{-3/2}\log^2{n}\sqrt{\frac{n^3}{\Delta}}\bigg{]}
    &\leq \sum_{t = 1}^n \mathbb{P}[K_t \geq 9\epsilon^{-1}\log{n}]\\
    &\leq \frac{n}{n^3} = \frac{1}{n^2}
\end{align*}
Taking both failure cases into account, we can see $O(\epsilon^{-3/2}\log^2{n} \sqrt{\frac{n^3}{\Delta}})$ queries are made and that a valid $(1 + \epsilon)\Delta$-coloring is obtained with probability at least $1 - \Theta(n^{-2})$ for sufficiently large $n$ and this can be boosted to $1 - \Theta(n^{-k})$ by altering the constants used for any $k > 0$.
\end{proof}
Now, note that whenever $\Delta \geq n^{1/3}\epsilon^{-1}$ this algorithm makes $\tilde{O}(\epsilon^{-1}n^{4/3})$ queries and the greedy algorithm suffices whenever $\Delta \leq n^{1/3}$.
\end{proof}

\section{Lower Bounds}
In this section, we establish some lower bounds for coloring and maximal matching with a very simple argument relying on the optimality of Grover's algorithm \cite{tight_bounds}:
\begin{theorem}
$\Omega(n$) queries are necessary to obtain an $O(\Delta)$- coloring on graphs with exactly one edge using only (quantum) adjacency matrix queries.
\end{theorem}
\begin{proof}[Proof of theorem 12]
 Finding a $2$-coloring for these graphs, given the promise that there is exactly one edge is the same as determining for which $i, j \in [n]$ with $i \not = j$ we have $M[i, j] = 1$. Clearly, we can see this as an instance of unstructured search with exactly one marked element among the $\binom{n}{2} = O(n^2)$ possible edges. Since unstructured search requires $\Omega(\sqrt{N})$ quantum queries, $\Delta + 1$ coloring requires $\Omega(n)$ quantum queries in general. This also implies that $\Omega(n)$ queries are necessary for $2\Delta$ coloring on general graphs since finding a $2-$coloring for graphs with exactly one edge is also equivalent to finding a $2\Delta$ coloring ($\Delta = 1$).\end{proof}
 It would be very interesting to see if more advanced techniques could yield a tighter lower bound, especially one which includes classical neighbor queries. We suspect that when all types of queries are available the same bound holds.

\appendix
\section{Details of \textbf{Find-Conflict}}
\label{appendix:Appendix A}
Here, we will give an explicit description of the quantum subroutine \textbf{Find-Conflict} for completeness. Recall that for a given set $S \subset V$ and some $v \in V$ we wish to determine if there exists $u \in S$ such that $(u, v) \in E$ with probability at least $\frac{2}{3}$. This algorithm is described in \cite{tight_bounds}.

\begin{algorithm}[H]
\SetAlgoLined
Let $U_O$ be the unitary oracle transformation\\
\textbf{for} $t = 1, 2:$\\
   \ \ \ \  $|\phi\rangle \gets \frac{1}{\sqrt{|S|}}\sum_{u \in S} |uv \rangle$\\
    \ \ \ \  $U_{\phi} = 2|\phi\rangle \langle \phi | - I$\\
   \ \ \ \  Apply $U_O$ and then $U_{\phi}$ to $|\phi \rangle$ each $\frac{\pi}{4}\sqrt{\frac{N}{t}}$ times\\
   \ \ \ \ Observe $| \phi \rangle$ resulting in candidate edge $(u, v)$\\
   \ \ \ \ Make the classical edge query $(u, v)$\\
   \ \ \ \ \textbf{if} $M[u, v] = 1$:\\
   \ \ \ \ \ \ \ \ \textbf{return} \textbf{True}\\
 \textbf{return} \textbf{False}
\caption{\textbf{Find-Conflict}$(v, S)$}
\end{algorithm}
\begin{proof}
 If there are exactly $1$ or $2$ conflicting neighbors for $v$ in $S$ then $u$ will be one of them with high probability when $t = 1$ or $t = 2$ respectively. When there are no conflicts the algorithm will always return \textbf{False}. The only worrisome case is when there are $3$ or more conflicts for $v$ in $S$. However, when there are $t_0 > 2$ conflicts each of the searches with $t = 1, 2$ will still succeed with probability at least $1/2$, resulting in an error probability of at most $\frac{1}{4}$. $O(\sqrt{|S|})$ queries are used every iteration, so $O(\sqrt{|S|})$ queries are used in total. See Theorem 3 of \cite{tight_bounds} for more details and \cite{Boyer_98} for general reference.
\end{proof}

\newpage

\bibliographystyle{plainurl}
\bibliography{sample}

\end{document}